\documentclass[oneside,reqno,english]{amsart}
\pdfoutput=1
\usepackage{lmodern}
\usepackage{berasans}
\usepackage{beramono}

\usepackage[T1]{fontenc}
\usepackage[latin9]{inputenc}
\usepackage{xcolor}
\usepackage{babel}
\usepackage{prettyref}
\usepackage{units}
\usepackage{mathrsfs}
\usepackage{enumitem}
\usepackage{amscd}
\usepackage{changebar}
\usepackage{amsbsy}
\usepackage{amstext}
\usepackage{amsthm}
\usepackage{amssymb}
\usepackage{graphicx}
\usepackage{setspace}
\PassOptionsToPackage{normalem}{ulem}
\usepackage{ulem}
\usepackage{microtype}
\usepackage[pdfusetitle,
 bookmarks=true,bookmarksnumbered=false,bookmarksopen=false,
 breaklinks=false,pdfborder={0 0 1},backref=false,colorlinks=true]
 {hyperref}
\hypersetup{
 citecolor=blue,linkcolor=magenta,anchorcolor=green}

\makeatletter

\providecolor{lyxadded}{rgb}{0.098,0.29,0.6}
\providecolor{lyxdeleted}{rgb}{1,0,0}
\DeclareRobustCommand{\mklyxadded}[1]{\textcolor{lyxadded}\bgroup#1\egroup}
\DeclareRobustCommand{\mklyxdeleted}[1]{\textcolor{lyxdeleted}\bgroup\mklyxsout{#1}\egroup}
\DeclareRobustCommand{\mklyxsout}[1]{\ifx\\#1\else\sout{#1}\fi}

\numberwithin{equation}{section}
\theoremstyle{plain}
\newtheorem{thm}{\protect\theoremname}
\theoremstyle{plain}
\newtheorem{lem}{\protect\lemmaname}

\@ifundefined{date}{}{\date{}}
\usepackage{upgreek}
\usepackage{urwchancal}
\usepackage{mathtools}
\usepackage{tikz-cd}
\usepackage{gensymb} 
\AtBeginDocument{%
  \setlength\abovedisplayskip{3pt plus 1pt minus 1pt}
  \setlength\belowdisplayskip{4pt plus 1pt minus 1pt}}
\newrefformat{cor}{Corollary \ref{#1}}
\newrefformat{prop}{Proposition \ref{#1}}
\newrefformat{conj}{Conjecture \ref{#1}}
\newrefformat{nthm}{Theorem \nameref{#1}}
\newrefformat{def}{Definition \ref{#1}}
\newrefformat{eq}{Eq.(\ref{#1})}
\newrefformat{rem}{Remark \ref{#1}}

\makeatother

\providecommand{\lemmaname}{Lemma}
\providecommand{\theoremname}{Theorem}

\begin{document}
\global\long\def\set#1#2{\left\{  #1\, |\, #2\right\}  }%

\global\long\def\cyc#1{\mathbb{Q}\!\left[\zeta_{#1}\right]}%

\global\long\def\mat#1#2#3#4{\left(\begin{array}{cc}
#1 & #2\\
#3 & #4
\end{array}\right)}%

\global\long\def\Mod#1#2#3{#1\equiv#2\, \left(\mathrm{mod}\, \, #3\right)}%

\global\long\def\inv{^{\,\textrm{-}1}}%

\global\long\def\pd#1{#1^{+}}%

\global\long\def\sym#1{\mathbb{S}_{#1}}%

\global\long\def\fix#1{\mathtt{Fix}\!\left(#1\right)}%

\global\long\def\map#1#2#3{#1\!:\!#2\!\rightarrow\!#3}%

\global\long\def\Map#1#2#3#4#5{\begin{split}#1:#2  &  \rightarrow#3\\
 #4  &  \mapsto#5 
\end{split}
 }%

\global\long\def\fact#1#2{#1\slash#2}%

\global\long\def\Gal#1{\mathtt{Gal}\!\left(#1\right)}%

\global\long\def\fixf#1{\mathbb{Q}\!\left(#1\right)}%

\global\long\def\gl#1#2{\mathsf{GL}_{#2}\!\left(#1\right)}%

\global\long\def\SL{\mathrm{SL}_{2}\!\left(\mathbb{Z}\right)}%

\global\long\def\zn#1{\left(\mathbb{Z}/\!#1\mathbb{Z}\right)^{\times}}%

\global\long\def\sn#1{\mathbb{S}_{#1}}%

\global\long\def\aut#1{\mathrm{Aut\mathit{\left(#1\right)}}}%

\global\long\def\FA#1{\vert#1\vert}%

\global\long\def\FB#1{\mathtt{Z}^{2}\!(#1)}%

\global\long\def\FC#1{#1^{{\scriptscriptstyle \flat}}}%

\global\long\def\FD#1{#1^{{\scriptscriptstyle \times}}}%

\global\long\def\FF#1#2{\mathrm{Fix}_{#1}\left(#2\right)}%

\global\long\def\FI#1{#1_{{\scriptscriptstyle \pm}}}%

\global\long\def\cl#1{\mathscr{C}\negthinspace\ell\!\left(#1\right)}%

\global\long\def\bl#1{\mathcal{B}\ell\!\left(#1\right)}%

\global\long\def\cft{\mathscr{C}}%

\global\long\def\gal#1{\upsigma_{#1}}%

\global\long\def\galpi#1{\boldsymbol{\uppi}{}_{#1}}%

\global\long\def\ann#1{\ker\mathfrak{#1}}%

\global\long\def\cent#1#2{\mathtt{C}_{#1}\!\left(#2\right)}%

\global\long\def\ch#1{\boldsymbol{\uprho}_{#1}}%

\global\long\def\qd#1{\mathtt{d}_{#1}}%

\global\long\def\cw#1{\mathtt{h}_{#1}}%

\global\long\def\tcl{\textrm{trivial class}}%

\global\long\def\ccl{\mathtt{C}}%

\global\long\def\zcl{\mathtt{z}}%

\global\long\def\om#1{\omega\!\left(#1\right)}%

\global\long\def\cs#1{\boldsymbol{\upmu}(\mathfrak{\mathfrak{#1}})}%

\global\long\def\irr#1{\mathtt{Irr}\!\left(#1\right)}%

\global\long\def\fc{\textrm{FC set}}%

\global\long\def\v{\mathtt{{\scriptstyle 0}}}%

\global\long\def\fm{\mathtt{N}}%

\global\long\def\du#1{\mathfrak{#1}^{{\scriptscriptstyle \perp}}}%

\global\long\def\lat{\mathcal{\mathscr{L}}}%

\global\long\def\zent#1{\mathtt{Z}(\mathfrak{#1})}%

\global\long\def\cech#1#2{\boldsymbol{\varpi}_{#1}\!\left(#2\right)}%

\global\long\def\idch{\boldsymbol{\mathfrak{1}}}%

\global\long\def\usub#1{\mathbf{\boldsymbol{\cup}}#1}%

\global\long\def\zquot#1#2{\mathfrak{#1}/\!#2}%

\global\long\def\extclass#1{#1\ccl}%

\global\long\def\dg#1{\hat{#1}}%

\global\long\def\cov#1{\mathfrak{#1}^{\intercal}}%

\global\long\def\vfc{\mathfrak{o}}%

\global\long\def\sc{\mathcal{\cov{\vfc}}}%

\global\long\def\res#1#2{\mathtt{res}_{#2}^{#1}}%

\global\long\def\clos#1{\langle#1\rangle}%

\global\long\def\mzq#1{\boldsymbol{\partial}\mathfrak{#1}}%

\global\long\def\ZD#1{\du{\left(\mzq{\mathfrak{#1}}\right)}}%

\global\long\def\ZR#1#2{\mathtt{Z}(\mathfrak{#1}\vert\mathfrak{#2})}%

\global\long\def\IR#1#2{\mathtt{I}\!\left(\mathfrak{#1}\vert\mathfrak{#2}\right)}%

\global\long\def\ext#1#2{\mathfrak{#1}\!\propto\!\mathfrak{#2}}%

\global\long\def\mzn#1#2{\boldsymbol{\partial}^{\mathtt{#1}}\mathfrak{#2}}%

\global\long\def\cmp#1{\textrm{(cf. #1 of [2])}}%

\global\long\def\ma{\mathrm{mass}}%
\global\long\def\fcs{\textrm{FC set}}%

\global\long\def\ke#1{\mathtt{ker}(#1)}%
\global\long\def\im#1{\mathtt{im}(#1)}%
\global\long\def\coker#1{\mathtt{coker}(#1)}%

\title{Exact sequences and the combinatorics of conformal models}
\author{P. Bantay\\
Institute for Theoretical Physics, E\"{o}tv\"{o}s Lor{\' a}nd University}
\begin{abstract}
We investigate the mutual relations between the centers of different
elements in the deconstruction lattice of a 2D conformal model, and
show how these can be described using exact sequences of abelian groups.
In particular, we exhibit a long exact sequence connecting the centers
of higher central quotients. 
\end{abstract}

\maketitle

\section{Introduction}

It has been understood for a long time that a substantial amount
of information about a rational conformal model \cite{BPZ,DiFrancesco-Mathieu-Senechal}
is encoded in its fusion algebra, which describes the possible couplings
between primary fields (i.e. which primaries, together with their
conformal descendants, could appear in the operator product of any
two given primaries). A famous result in this direction, instigating
important mathematical developments over the years, is the celebrated
formula of Verlinde \cite{Verlinde1988} relating the fusion rules
of the model, i.e. the structure constants of the fusion algebra,
to the matrix describing the transformation properties of the chiral
characters under the modular transformation $\map S{\tau}{\nicefrac{-1}{{\textstyle \tau}}}$,
and which allows to reconstruct the latter from the knowledge of the
fusion rules and the conformal weights of the primaries. But it is
fair to say that this is just the tip of the iceberg, for several
similar relations are known, e.g. for the Frobenius-Schur indicators
of the primaries \cite{Bantay1997a}, or for the traces of finite
order mapping classes \cite{Bantay2003c}. 

As has been discussed in \cite{Bantay2020a,Bantay2021}, an interesting
structure related to the fusion algebra results from the consideration
of those collections of primaries that, besides containing the vacuum,
are closed under the fusion product. These may be shown to form a
lattice with many nice properties, like being self-dual and modular
\cite{Crawley1973,Gratzer2011}. Some elements of this so-called deconstruction
lattice (the local ones) may be identified with the representation
rings of suitable finite groups, allowing the transfer to them of
several standard notions from group theory (like commutativity, nilpotency,
 etc.) which turn out to make perfectly good sense for generic elements. 

In particular, there is a way to define the notions of center and
that of central quotients/extensions \cite{Bantay2020a}, which correspond,
in case of local elements, to the standard group theoretic notions.
The importance of all this stems from the fact that there is a well
understood relationship between properties of central quotients and
extensions, and this can facilitate greatly the analysis of specific
models. A notable exception to this, in complete analogy with the
case of groups, concerns the structure of the center itself, as there
is no obvious connection between the centers of central quotients/extensions.
The aim of the present note is to show that a useful characterization
of this relationship can be nevertheless given in terms of exact sequences
of abelian groups \cite{Robinson}. 

In order to be accessible to a wider readership, we briefly summarize
background material on the deconstruction lattice in \prettyref{sec:The-deconstruction-lattice},
on central quotients and extensions in \prettyref{sec:Central-quotients-and},
and on exact sequences in \prettyref{sec:Exact-sequences}. Then we
turn to our main subject, and study in \prettyref{sec:basics} the
restriction homomorphism that connects the centers of different elements
in the deconstruction lattice. The most important results can be found
in \prettyref{sec:Galois}, which investigates the Galois correspondence
between central quotients and subgroups of the center, resulting in
exact sequences able to describe the subtleties of this relationship.
The theme of \prettyref{sec:Long-exact-sequences} is a long exact
sequence connecting the centers of higher quotients that proves useful
in actual computations, while \prettyref{sec:Extensions-vs-quotients}
is concerned with the case of central extensions instead of quotients,
based on the dual nature of these two notions. In the last section
we summarize the results, and comment on open questions and possible
future developments. Finally, as the modularity of the deconstruction
lattice plays a pivotal role in some of the results of \prettyref{sec:Galois},
a streamlined proof of this fundamental result is presented in the
\nameref{sec:Appendix}, based on an explicit characterization of
the join operation that could prove interesting in itself.

\section{The deconstruction lattice\label{sec:The-deconstruction-lattice}}

Let's consider a 2D rational conformal model \cite{BPZ,DiFrancesco-Mathieu-Senechal}
with a finite number of primaries. We shall denote by $N_{pq}^{r}$
the fusion rule coefficients, that is the multiplicity of a primary
$r$ in the fusion product of the primaries $p$ and $q$. The collection
of all those subsets of primaries that contain the vacuum and are
closed under the fusion product (meaning that if $N_{pq}^{r}\!>\!0$
with both $p$ and $q$ belonging to it, then $r$ does also belongs
to it) may be shown \cite{Bantay2020a} to form a modular lattice
$\lat$, termed the \emph{deconstruction lattice} because of its fundamental
role in the classification of the different orbifold deconstructions
of the model \cite{Bantay2019a,Bantay2020}. The ordering in $\lat$
is simply set inclusion, and the meet operation is set intersection
(the join operation being less obvious, but see the \nameref{sec:Appendix}).

Recall \cite{Bantay2020a} that to each $\mathfrak{g\!\in\!\lat}$
one can assign a partition of the primaries of the model into so-called
$\mathfrak{g}$-classes characterized by the fact that the irreducible
representations of the fusion algebra corresponding to different elements
in the same class coincide when restricted to the elements of $\mathfrak{g}$.
Of utmost importance is the $\mathfrak{g}$-class that contains the
vacuum, the\emph{ trivial class} $\du g$, which may be shown to be
itself an element of $\lat$. It is straightforward that $\du g\!\subseteq\!\du h$
whenever $\mathfrak{h}\!\subseteq\!\mathfrak{g}$, and that the trivial
class of $\du g$ is $\mathfrak{g}$ itself, hence the lattice $\lat$
is \emph{self-dual}, i.e. endowed with an order-reversing and involutive
\emph{duality map} $\mathfrak{g}\!\mapsto\!\du g$ that relates the
join and meet operations via $\du{\left(\mathfrak{g}\!\vee\!\mathfrak{h}\right)}\!=\!\du g\!\cap\du h$.

As it turns out, self-dual lattices are closely related to undirected
graphs (with possible loops): any such graph determines a self-dual
lattice, and any self-dual lattice comes from a suitable graph\footnote{Note that quite different graphs may lead to the same lattice, but
the collection of all such graphs may be characterized in a simple
manner \cite{Bantay2022}.}. A fairly non-intuitive result is that the deconstruction lattice
of a conformal model corresponds to its locality graph, i.e. the graph
whose vertices are the primary fields, with two of them adjacent if
they are mutually local, i.e. if their OPE is single-valued. Actually,
instead of the locality graph one can use the so-called locality diagram,
whose vertices correspond to equilocality classes of primaries, i.e.
collections of primaries that are mutually local with the same set
of primaries, providing a nice graphical representation of the deconstruction
lattice \cite{Bantay2021a,Bantay2022}.

Especially important are those elements $\mathfrak{g\!\in\!\lat}$,
termed \emph{local} ones, for which $\mathfrak{g}\!\subseteq\!\du g$,
since these provide the input data for the orbifold deconstruction
procedure \cite{Bantay2019a,Bantay2020}, and correspond to the different
orbifold realizations of the given model. In particular, for each
local $\mathfrak{g\!\in\!\lat}$ there exists a finite group, the
twist group of the corresponding orbifold, whose representation ring
(viewed as a $\lambda$-ring, i.e. taking into account the different
possible symmetrizations of tensor powers) is isomorphic\footnote{Note that this does not fix the group uniquely, because of the existence
of so-called Brauer-tuples, i.e. non-isomorphic groups with identical
character tables and power maps \cite{Lux-Pahlings}. But such exceptions
are rather sparse, and the resulting ambiguity can be handled.} with $\mathfrak{g}$. What is more, if $\mathfrak{g\!\in\!\lat}$
is local and $\mathfrak{h\!\in\!\lat}$ is contained in it, then $\mathfrak{h}$
is local too, and the group corresponding to $\mathfrak{h}$ is a
homomorphic image (i.e. factor group) of the one corresponding to
$\mathfrak{g}$. An important feature of local $\mathfrak{g\!\in\!\lat}$
is that all of their elements, besides having (rational) integer quantum
dimension, have either integer or half-integer conformal weight.

Let us note that in case of abelian models, when all the primaries
are simple currents (primaries of quantum dimension 1) \cite{SY1,Intriligator,SY2},
the fusion closed sets are nothing but the different subgroups of
the group of simple currents, hence the deconstruction lattice is
isomorphic to the subgroup lattice of the latter, with the duality
map being determined by the relation of mutual locality, which is
in turn determined by the distribution of conformal weights. In such
a case a simple current extension by integer spin simple currents
\cite{Fuchs1996a} has the same result as the corresponding orbifold
deconstruction of the model\footnote{More generally, simple current extensions of arbitrary conformal models
can be understood as orbifold deconstructions for $\mathfrak{g}\!\in\!\lat$
consisting purely of integer spin simple currents.}.

It should transpire from the above that the deconstruction lattice
$\lat$ is an important algebraic/combinatorial invariant of the
conformal model under study. In particular, it gives a concise description
of the different realizations of the given model as an orbifold of
another one. Understanding the structure of $\lat$ through the relation
of its elements to each other could provide information about the
model that would be hard to obtain by other means. One such piece
of information is related to the notion of central quotients and extensions,
to which we now turn.

\section{Central quotients and extensions\label{sec:Central-quotients-and}}

Given $\mathfrak{g}\!\in\!\lat$, an important numerical characteristic
of any $\mathfrak{g}$-class $\mathcal{\ccl}$ is its \emph{mass}
\begin{equation}
\cs{\ccl}\!=\!\sum_{p\in\ccl}\qd p^{2}\label{eq:massdef}
\end{equation}
where $\qd p$ denotes the quantum dimension of the primary $p$.
An especially important role is played by those $\mathfrak{g}$-classes
whose $\ma$ is minimal (equal to that of the trivial class $\du g$),
the \emph{central classes}, which may be shown \cite{Bantay2020a}
to form an abelian group $\zent{\mathfrak{g}}$, the \emph{center}
of $\mathfrak{g}$, that permutes the collection of all $\mathfrak{g}$-classes,
and whose identity element is the trivial class. For local $\mathfrak{g}\!\in\!\lat$
corresponding to the finite group $G$, the center $\zent{\mathfrak{g}}$
is isomorphic with the center of $G$. To each central class $\zcl\!\in\!\zent g$
one may associate a complex-valued function $\map{\varpi_{\zcl}}{\mathfrak{g}}{\mathbb{C}}$,
its \emph{central characte}r, such that for $\alpha,\beta\!\in\!\mathfrak{g}$
one has $\varpi_{\zcl}\!\left(\gamma\right)\!=\!\varpi_{\zcl}\!\left(\alpha\right)\varpi_{\zcl}\!\left(\beta\right)$
whenever $N_{\alpha\beta}^{\gamma}\!>\!0$, and one may show that
the map that assigns to each central class its central character provides
an explicit isomorphism between $\zent g$ and its Pontryagin dual.

It turns out \cite{Bantay2020a} that each subgroup $H\!<\!\zent g$
determines another element of $\lat$, the \emph{central quotient}
$\zquot g{\mathit{H}}$, and there is a one-to-one correspondence
between elements $\mathfrak{h}\!\in\!\lat$ that satisfy $\zquot g{\mathit{H}}\!\subseteq\!\mathfrak{h}\!\subseteq\!\mathfrak{g}$
and subgroups of $H$. In particular, each $\mathfrak{g\!\in\!\lat}$
has a \emph{maximal central quotient}
\begin{equation}
\mzq{\mathfrak{g}}=\zquot g{\mathit{~\zent{\mathfrak{g}}}}\label{eq:mzqdef}
\end{equation}
that is contained in all central quotients, and each $\mathfrak{h}\!\in\!\lat$
that satisfies $\mzq{\mathfrak{g}}\!\subseteq\!\mathfrak{h}\!\subseteq\!\mathfrak{g}$
is a central quotient of $\mathfrak{g}$, i.e. $\mathfrak{h}\!=\!\zquot g{\mathit{H}}$
for some subgroup $H\mathfrak{<}\zent g$. Actually, the relevant
subgroup reads $H\!=\!\set{\zcl\!\in\!\zent g}{\zcl\!\subseteq\!\du h}$,
and this implies that the union of the central classes contained in
$H$ (viewed as collections of primaries) equals $\du{\mathfrak{h}}$,
i.e. $\usub H\!=\!\du h$. In particular,
\begin{equation}
\usub{\zent g\!=\!\du{\left(\mzq g\right)}}\label{eq:usubzent}
\end{equation}

The other way round, a \emph{central extension} of $\mathfrak{g\!\in\!\lat}$
is an element $\mathfrak{h}\!\in\!\lat$ of which $\mathfrak{g}$
is a central quotient. Using the duality map $\mathfrak{g}\!\mapsto\!\du g$
of $\lat$, one can show \cite{Bantay2020a} that there is a one-to-one
correspondence between central extensions of $\mathfrak{g}$ and subgroups
of $\zent{\du g}$. In particular, there is a \emph{maximal central
extension} $\cov{\mathfrak{g}}\!=\!\du{\left(\mzq{\left(\du{\mathfrak{g}}\right)}\right)}$
that contains all central extensions of $\mathfrak{g}$, and each
$\mathfrak{h}\!\in\!\lat$ that satisfies $\mathfrak{g}\!\subseteq\!\mathfrak{h}\!\subseteq\!\cov{\mathfrak{g}}$
is a central extension of $\mathfrak{g}$, i.e. $\mathfrak{g}$ is
a central quotient of $\mathfrak{h}$; put another way, the inclusion
$\mzq{\mathfrak{g}}\!\subseteq\!\mathfrak{h}\!\subseteq\!\mathfrak{g}$
is equivalent to $\mathfrak{h}\!\subseteq\!\mathfrak{g}\!\subseteq\!\cov{\mathfrak{h}}$.

The map that assigns to an element $\mathfrak{g}\!\in\!\lat$ its
maximal central quotient (resp. extension) $\mzq g$ (resp. $\cov g$)
is order-preserving, i.e. $\mzq h\!\subseteq\!\mzq g$ (resp. $\cov h\!\subseteq\!\cov g$)
whenever $\mathfrak{h}\!\subseteq\!\mathfrak{g}$. What is more interesting
is the fact that one has
\begin{equation}
\mzq{\left(\cov g\right)\!\subseteq\!\mathfrak{g}\!\subseteq\!\cov{\left(\mzq g\right)}}\label{eq:ineq1}
\end{equation}
for every $\mathfrak{g}\!\in\!\lat$. Indeed, since $\mathfrak{g}$
is a central quotient of $\cov g$, it should contain the maximal
central quotient of the latter, and in the same vein, since $\mzq{\mathfrak{g}}$
is a central quotient of $\mathfrak{g}$, the latter is a central
extension of $\mzq{\mathfrak{g}}$, consequently it has to be contained
in $\cov{\left(\mzq g\right)}$. Actually, one may show that the map
$\mathfrak{g}\!\mapsto\!\cov{\left(\mzq g\right)}$ is a closure operation
on $\lat$, whose fixed points are those elements of $\lat$ that
are themselves maximal central extensions, i.e. of the form $\cov g$
for some $\mathfrak{g}\!\in\!\lat$. 

The importance of central quotients and extensions stems from the
fact that their properties are determined to a large extent by group
theory. For example, the $\ma$ of the quotient $\zquot g{\mathit{H}}$
is $\FA H$ times that of $\mathfrak{g}$,
\begin{equation}
\cs{\zquot g{\mathit{H}}}\!=\!\FA H\cs g\label{eq:zqmass}
\end{equation}
and one has a fairly good description of the classes of the quotient
$\zquot g{\mathit{H}}$ in terms of the classes of $\mathfrak{g}$
and the action of $H\!<\!\zent{\mathfrak{g}}$ on them \cite{Bantay2020a}.
Based on this, one expects that the center of $\mathfrak{g\!\in\!\lat}$
and that of its quotients $\zquot g{\mathit{H}}$ are related somehow.
As we shall show in the sequel, such a relation does indeed exist,
but to describe it neatly we shall need the machinery of exact sequences
\cite{Robinson}.

As an illustration of the above, consider the minimal $N\!=\!2$ superconformal
model of central charge $c\!=\!2$. The Hasse diagram of its deconstruction
lattice is depicted in \prettyref{fig:sv2_4}, each element being
labeled by (the isomorphism type of) its center. Inspecting the figure,
one recognizes that it is made up of two sublattices related by the
duality map (which in this case is a flip, i.e. a $180\degree$ rotation
around the center of the diagram), each one isomorphic with the subgroup
lattice of the abelian group $\mathbb{Z}_{12}\!\times\!\mathbb{Z}_{2}$,
the latter being nothing but the group of simple currents\footnote{It is a general fact that the center of the maximal element of the
deconstruction lattice is isomorphic with the group of simple currents.}.

But the picture is usually much more complicated, as exemplified by
the $\mathbb{Z}_{2}$ orbifold of the compactified boson at radius
$R\!=\!6$. The Hasse diagram of its deconstruction lattice is depicted
in \prettyref{fig:at36}; once again, the duality map is the flip
around the center of the diagram, and the nodes are labeled by their
centers. While parts of the diagram mimic the subgroup lattice of
some abelian group, the overall pattern is clearly more complicated
than in the previous example. It was the desire to understand the
structure underlying these patterns that motivated the present work.
\begin{figure}
~~~~\includegraphics[scale=0.52]{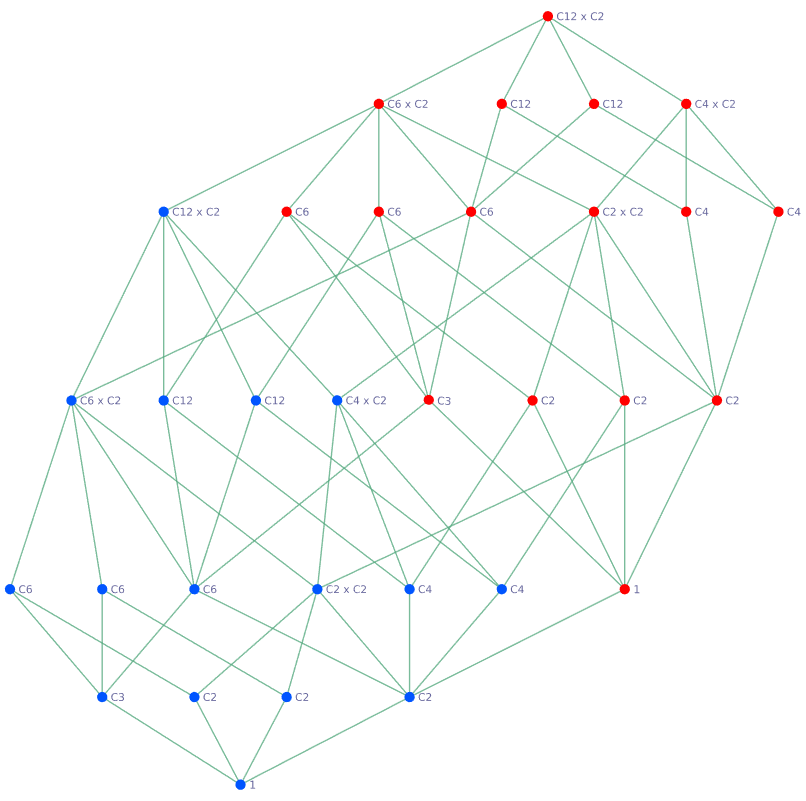}

\caption{\label{fig:sv2_4}Deconstruction lattice of the minimal $N\!=\!2$
superconformal model of central charge $c\!=\!2$, the nodes being
labeled by their centers (with $\mathtt{C}n$ denoting a cyclic group
of order $n$).}
\end{figure}

\begin{figure}[t]
\centering
\includegraphics[scale=0.32]{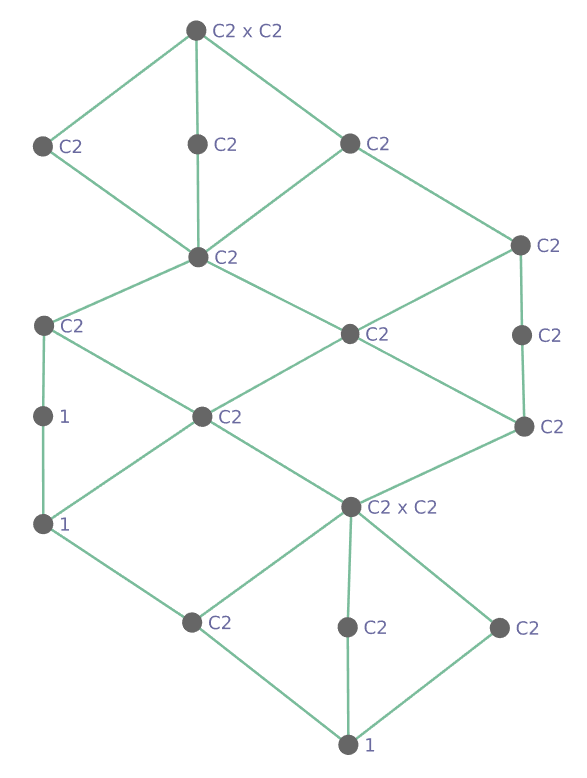}

\caption{\label{fig:at36}Deconstruction lattice of the $\mathbb{Z}_{2}$ orbifold
of the compactified boson at radius $R\!=\!6$, the nodes being labeled
by their centers.}
\end{figure}

Finally, let's remark that there is no reason to stop at the maximal
central quotient $\mzq g$, for one can also consider the maximal
central quotient of the latter, and so on repeatedly. Defining $\mzn{\mathit{k}}g$
for $\mathfrak{g}\!\in\!\lat$ and a positive integer $k$ via $\mzn 1g\!=\!\mzq{\mathfrak{g}}$
and the recursion $\mzn{\mathit{k}+1}g\!=\!\mzq{(\mzn{\mathit{k}}g)}$,
we get a descending chain
\begin{equation}
\mathfrak{g}\!\supseteq\!\mzn 1g\!\supseteq\cdots\!\supseteq\!\mzn{\mathit{n}}g\label{eq:ucs}
\end{equation}
of elements of $\lat$, an analogue of the upper central series from
group theory \cite{Robinson}. In case the lattice $\lat$ is finite,
the inclusions in \prettyref{eq:ucs} imply that the sequence stabilizes
after finitely many steps, i.e. there exists some integer $N$ such
that $\mzn{\mathit{N}-1}g\!\neq\!\mzn{\mathit{N}}g\!=\!\mzn{\mathit{N}+1}g$.
We shall denote by $\mathrm{\mzn{\mathit{\infty}}g}$ this last term
of the upper central series \prettyref{eq:ucs}, and call it the \emph{hypercentral
quotient}\footnote{This terminology stems from the fact that for local $\mathfrak{g}\!\in\!\lat$
that corresponds to the finite group $G$, the hypercentral quotient
$\mathrm{\mzn{\mathit{\infty}}g}$ is itself local, and corresponds
to the factor group of $G$ by its hypercenter \cite{Robinson}.} of $\mathfrak{g}$. Note that, because $\mathrm{\mzn{\mathit{\infty}}g}$
equals its own maximal central quotient by definition, it has trivial
center, $\zent{\mathrm{\mzn{\mathit{\infty}}g}}\!=\!1$, and this
means that any element of $\lat$ may be obtained by repeated central
extensions from one with trivial center. As the structure of central
extensions is well under control, this helps to reduce the study of
subtler aspects of $\lat$ to that of its centerless elements.

\section{Exact sequences\label{sec:Exact-sequences}}

As already alluded to previously, our goal is to understand the relations
between the centers of different elements of the deconstruction lattice,
in particular those that are central quotients or extensions of each
other, and a most efficient way to express these connections is through
the use of exact sequences. Since the latter are not part of the everyday
toolkit of physicists, let's briefly summarize those basics facts
about them that we shall need in what follows.

Recall \cite{Robinson} that an \emph{exact sequence} of (abelian)
groups is a sequence of group homomorphisms $\map{\phi_{i}}{A_{i-1}}{A_{i}}$
for $i\!=\!1,\ldots,n$, such that the kernel of $\phi_{i+1}$ equals
the image of $\phi_{i}$ for $1\!\leq\!i\!<\!n$; note that the domain
of $\phi_{i+1}$ equals the range of $\phi_{i}$. Their usual notation
is
\begin{equation}
\begin{CD}A_{0}@>\phi_{1}>>A_{1}@>\phi_{2}>>\cdots@>\phi_{n-1}>>A_{n-1}@>\phi_{n}>>A_{n}\end{CD}\label{eq:sq1}
\end{equation}
 In many cases one is not really interested in the actual homomorphisms
that connect the $A_{i}$, but only in the relationship between them
that follows from the existence of the exact sequence; in such cases
one usually drops any explicit reference to the $\phi_{i}$ (unless
the specification of some of them provides useful extra information).

The most ubiquitous are the so-called \emph{short exact sequence}s
\begin{equation}
\begin{CD}1@>>>A@>>>B@>>>C@>>>1\end{CD}\label{eq:sq2}
\end{equation}
where $1$ denotes the (isomorphism class of the) trivial group of
order 1. The homomorphism theorem \cite{Alperin-Bell,Robinson} tells
us at once that the existence of such an exact sequence is tantamount
to the existence of a subgroup $K\!<\!B$, the kernel of the homomorphism
from $B$ to $C$, which satisfies $K\!\cong\!A$ and $C\!\cong\!\fact{B\!}{\!K}$.

Next in line come the four-term exact sequences
\begin{equation}
\begin{CD}1@>>>A@>>>B@>>>C@>>>D@>>>1\end{CD}\label{eq:sq3}
\end{equation}
which play an important role e.g. in algebraic number theory \cite{Janusz},
and whose existence tells us that there exists a homomorphism $\map{\phi}BC$
whose kernel is isomorphic with $A$, and whose cokernel $\mathtt{coker}\!\left(\phi\right)\!=\!\fact C{\mathtt{im}}\left(\phi\right)$
is isomorphic with $D$ (note that, since every subgroup of an abelian
group is normal, the cokernel makes perfectly good sense for homomorphism
of abelian groups). In particular, it follows from Lagrange's theorem
\cite{Robinson} that the orders of the groups appearing in a four-term
exact sequence are related by $\FA A\FA C\!=\!\FA B\FA D$.

Finally, there are the l\emph{ong exact sequence}s
\begin{equation}
\begin{CD}1@>>>A_{1}@>>>A_{2}@>>>\cdots@>>>A_{n}@>>>1\end{CD}\label{eq:sq4}
\end{equation}
whose existence signals a more subtle and less immediate relation
between the $A_{i}$. Note that long exact sequences are equivalent
to a suitable collection of overlapping short exact sequences. To
see how this comes about, let's consider the exact sequences 
\begin{equation}
\begin{CD}\,1@>>>A_{1}@>>>\cdots@>\phi_{n-1}>>A_{n}@>\psi>>B@>>>1\end{CD}\label{eq:sq5}
\end{equation}
and
\begin{equation}
\minCDarrowwidth25pt\begin{CD}1@>>>B@>\iota>>A_{n+1}@>\phi_{n+1}>>\cdots@>>>A_{n+m}@>>>1\end{CD}\label{eq:sq6}
\end{equation}
with $B\neq1$ (should $B$ be trivial, the above sequences would
shorten and would not have a common anchor). Then, since $\psi$ is
surjective with $\ke{\psi}\!=\!\im{\phi_{n-1}}$, while $\iota$ is
injective with $\im{\iota}\!=\!\ke{\phi_{n+1}}$, the composite $\phi_{n}\!=\!\iota\!\circ\!\psi$
satisfies $\ke{\phi_{n}}\!=\!\ke{\psi}\!=\!\im{\phi_{n-1}}$ and $\im{\phi_{n}}\!=\!\im{\iota}\!=\!\ke{\phi_{n+1}}$,
hence the sequence
\begin{center}
\begin{tikzcd}
1\arrow[r]&A_1\arrow[d,phantom,""{coordinate,name=Z}]\arrow{r}&\cdots\arrow{r}&A_{n-1}\arrow{r}{\phi_{n-1}}&A_n
\arrow["\hspace{0.2 cm}\phi_n\hspace{2.5 cm}" swap,dlll,rounded corners,to path={--([xshift=3ex]\tikztostart.east)|-(Z)[near end]\tikztonodes-|([xshift=-4ex]\tikztotarget.west)--(\tikztotarget)}]\\ &A_{n+1}\arrow{r}{\phi_{n+1}}&\cdots\arrow{r}&A_{n+m}\arrow[r]&1\hspace{1 cm}
\end{tikzcd}
\par\end{center}

\begin{flushleft}
is exact. This trick allows us to combine exact sequences into longer
ones, or to decompose a long exact sequence into a collection of short
exact sequences. We shall encounter the application of these ideas
in the later sections.
\par\end{flushleft}

\section{The restriction homomorphism\label{sec:basics}}

\global\long\def\rex#1#2{\mathfrak{d}_{\mathfrak{{\scriptscriptstyle #1}}}^{\mathfrak{{\scriptscriptstyle #2}}}}%
\global\long\def\ex#1{#1^{\mathfrak{h}}}%

After all these preliminaries, it is time to focus on our basic problem,
namely the relation between the center of an element $\mathfrak{g}\!\in\!\lat$
and the centers of its different central quotients/extensions. As
it turns out, it is worth to consider the more general problem of
relating the center of $\mathfrak{g}\!\in\!\lat$ to that of any $\mathfrak{h}\!\in\!\lat$
contained in it. That this approach is meaningful is based on the
following two results from \cite{Bantay2020a}:
\begin{enumerate}
\item if $\mathfrak{h}\!\subseteq\!\mathfrak{g}$, then any $\mathfrak{h}$-class
is a union of $\mathfrak{g}$-classes;
\item an $\mathfrak{h}$-class that contains a central $\mathfrak{g}$-class
is itself central.
\end{enumerate}
It follows that for any $\mathfrak{g},\mathfrak{h}\!\in\!\lat$ with
$\mathfrak{h}\!\subseteq\!\mathfrak{g}$ there is a unique central
$\mathfrak{h}$-class that contains a given central $\mathfrak{g}$-class
$\zcl\!\in\!\zent g$. 
\begin{thm}
\label{thm:resker}For $\mathfrak{g},\mathfrak{h}\!\in\!\lat$ with
$\mathfrak{h}\!\subseteq\!\mathfrak{g}$, the map that assigns to
each central $\mathfrak{g}$-class the unique $\mathfrak{h}$-class
that contains it is a homomorphism $\map{\rex gh}{\zent g}{\zent h}$,
with\textup{\emph{ image}}\textup{ $\IR gh\!=\!\set{\zcl\!\in\!\zent h\!}{\zcl\!\cap\!\ZD g\!\neq\emptyset}$}\textup{\emph{
and}} kernel $\ZR gh\!=\!\set{\zcl\!\in\!\zent g}{\zcl\!\subseteq\!\du h}$\textup{. }
\end{thm}
\begin{proof}
Let's denote by $\ex{\zcl}$ the $\mathfrak{h}$-class containing
the central $\mathfrak{g}$-class $\zcl\!\subseteq\!\zent g$. Recall
from \prettyref{sec:Central-quotients-and} that to each central class
one can assign its central character, and this provides an explicit
isomorphism between the center and its Pontryagin dual. Because $\zcl\!\subseteq\!\ex{\zcl}$
and $\mathfrak{h}\!\subseteq\!\mathfrak{g}$ implies $\cech{\ex{\zcl}}{\alpha}\!=\!\cech{\zcl}{\alpha}$
for $\alpha\!\in\!\mathfrak{h}$, it follows that $\cech{\ex{(\zcl_{1}\zcl_{2})}}{\alpha}\!=\!\cech{\zcl_{1}\zcl_{2}}{\alpha}\!=\!\cech{\zcl_{1}}{\alpha}\cech{\zcl_{2}}{\alpha}\!=\!\cech{\ex{\zcl_{1}}}{\alpha}\cech{\ex{\zcl_{2}}}{\alpha}$
for $\zcl_{1},\zcl_{2}\!\in\!\zent g$. Since a central class is
uniquely determined by its central character, this shows that $\ex{(\zcl_{1}\zcl_{2})}\!=\!\ex{\zcl_{1}}\ex{\zcl_{2}}$,
proving that $\rex gh$ is indeed a homomorphism. Now $\zcl\!\in\!\ke{\rex gh}$,
i.e. $\ex{\zcl}\!=\!\du h$ precisely when $\zcl\!\subseteq\!\du h$,
and this proves that $\ke{\rex gh}\!=\!\ZR gh$. But a central $\mathfrak{h}$-class
belongs to the image of $\rex gh$ in case it contains at least one
central $\mathfrak{g}$-class, i.e. if its intersection with $\usub{\zent g}$
is not empty, one finally concludes that $\im{\rex gh}\!=\!\set{\zcl\!\in\!\zent h\!}{\zcl\cap\ZD g\!\neq\emptyset}\!=\!\IR gh$
by using \prettyref{eq:usubzent}. 
\end{proof}
Notice that one has $\ZR gg\!=\!1$ and $\IR gg\!=\!\zent g$ as a
consequence of \prettyref{eq:usubzent}, while $\ZR gh\!=\!\zent g$
and $\IR gh\!=\!1$ for $\mathfrak{h\!\subseteq\!\mzq g}$, and in
particular $\ZR g{\mzq g}\!=\!\zent g$. Moreover, it follows from
\prettyref{thm:resker} and the homomorphism theorem \cite{Robinson}
that
\begin{equation}
\IR gh\!\cong\!\fact{\zent g\!}{\ZR gh}\label{eq:seq1}
\end{equation}
and by Lagrange's theorem this implies that the order of $\IR gh$
(being a homomorphic image, resp. subgroup) must divide both the order
of $\zent g$ and that of $\zent h$. In particular, if the order
of $\zent g$ and of $\zent h$ are coprime, then $\IR gh\!=\!1$,
hence $\ZR gh\!=\!\zent g$, and this implies $\ZD g\!=\!\usub{\zent g}\!\subseteq\!\du h$,
i.e. $\mathfrak{h\!\subseteq\!\mzq g}$.

\global\long\def\jj{\mathfrak{h}}%
\global\long\def\hh{\mathfrak{j}}%
\global\long\def\inn#1#2{\lat\!\left(\mathfrak{#2},\!\mathfrak{#1}\right)}%
\global\long\def\zb{\overline{\zcl}}%

We shall call $\map{\rex gh}{\zent g}{\zent h}$ the \emph{restriction
homomorphism}. Clearly, for $\mathfrak{h}\subseteq\mathfrak{j}\subseteq\mathfrak{g}$
the restriction homomorphisms $\rex g{\hh}$ and $\rex{\hh}{\jj}$
can be composed, and this observation leads to the following basic
result.
\begin{lem}
\label{lem:comp} $\ZR g{\hh}<\ZR g{\jj}$ and $\IR g{\jj}<\IR{\hh}{\jj}$
in case $\mathfrak{h}\subseteq\mathfrak{j}\subseteq\mathfrak{g}$,
and the diagram

~~~~~~~~~~~~~~~~~~~~~~~~~~~~\begin{tikzcd}
\zent{g}\arrow[swap]{rd}{\rex{g}{\hh}} \arrow[swap]{rr}{\rex{g}{\jj}}&&\zent{\jj}\\&\zent{\hh}\arrow[swap]{ru}{\rex{\hh}{\jj}}&
\end{tikzcd}
\begin{flushleft}
is commutative, hence $\rex{\hh}{\jj}\!\left(\IR g{\hh}\right)=\IR g{\jj}$
and $\rex g{\hh}\!\left(\ZR g{\jj}\right)=\IR g{\hh}\cap\ZR{\hh}{\jj}$.
\par\end{flushleft}
\end{lem}
\begin{proof}
$\mathfrak{h}\!\subseteq\!\mathfrak{j}\!\subseteq\!\mathfrak{g}$
implies $\du j\!\subseteq\!\du h$ and $\du{\left(\mzq g\right)}\!\subseteq\!\du{\left(\mzq j\right)}$,
hence $\zcl\!\in\!\ZR g{\hh}$ implies $\zcl\!\in\!\ZR gh$, while
$\zcl\!\in\!\IR g{\jj}$ gives $\zcl\!\in\!\IR j{\jj}$, proving the
two inclusions. The commutativity of the diagram, i.e. the equality
$\rex{\hh}{\jj}\!\circ\!\rex g{\hh}\!=\!\rex g{\mathfrak{h}}$ is
immediate from the definitions, and it implies at once the equality
$\rex{\hh}{\jj}\!\left(\IR g{\hh}\right)\!=\!\rex{\hh}{\jj}\!\left(\im{\rex g{\mathfrak{j}}}\right)\!=\!\im{\rex{\hh}{\jj}\!\circ\!\rex g{\hh}}\!=\!\im{\rex g{\mathfrak{h}}}\!=\!\IR g{\jj}$.
To prove the last claim, note that in case $\zcl\!\in\!\rex g{\hh}\!\left(\ZR g{\jj}\right)$
there exists $\zb\!\in\!\ZR g{\jj}$ such that $\zcl\!=\!\rex g{\hh}\!\left(\zb\right)$,
hence $\rex{\hh}{\jj}\!\left(\zcl\right)\!=\!\rex g{\mathfrak{h}}\!\left(\zb\right)\!=\!\du h$,
i.e. $\zcl\!\in\!\ZR{\hh}{\jj}$, and this shows that $\rex g{\hh}\!\left(\ZR g{\jj}\right)\!<\!\ZR{\hh}{\jj}$;
as the inclusion $\rex g{\hh}\!\left(\ZR g{\jj}\right)\!<\!\im{\rex g{\hh}}\!=\!\IR g{\hh}$
is immediate, we conclude that $\rex g{\hh}\!\left(\ZR g{\jj}\right)\!\subseteq\!\IR g{\hh}\!\cap\!\ZR{\hh}{\jj}$.
On the other hand, $\zcl\!\in\!\IR g{\hh}\cap\ZR{\hh}{\jj}$ if there
exists $\zb\!\in\!\zent g$ such that $\zcl\!=\!\rex g{\hh}\!\left(\zb\right)$
belongs to $\ZR{\hh}{\jj}\!=\!\ke{\rex{\hh}{\jj}}$, and this happens
precisely when $\rex g{\jj}\!\left(\zb\right)\!=\!\left(\rex{\hh}{\jj}\!\circ\!\rex g{\hh}\right)\!\left(\zb\right)\!=\!\rex{\hh}{\jj}\!\left(\zcl\right)\!=\!\du h$,
i.e. $\zb\!\in\!\ke{\rex g{\jj}}$, which implies $\zcl\!\in\!\rex g{\hh}\!\left(\ZR g{\jj}\right)$,
proving finally that $\rex g{\hh}\!\left(\ZR g{\jj}\right)\!\supseteq\!\IR g{\hh}\cap\ZR{\hh}{\jj}$.
\end{proof}
Before the next result, let's recall the isomorphism theorems of group
theory \cite{Robinson}: for any two subgroups $A$ and $B$ of an
abelian group $C$, one has
\begin{equation}
\fact A{A\cap B}\cong\fact{AB}B\label{eq:1stiso}
\end{equation}
Moreover, in case $A\!<\!B$ one has
\begin{equation}
\fact{\left(\fact CA\right)\!}{\!\left(\fact BA\right)\!\cong\!\fact CB}\label{eq:2ndiso}
\end{equation}

\begin{lem}
\label{lem:compcor}For $\mathfrak{h}\!\subseteq\!\mathfrak{j}\!\subseteq\!\mathfrak{g}$
one has isomorphisms
\begin{align}
\IR g{\hh}\!\cap\!\ZR{\hh}{\jj} & \cong\fact{\ZR g{\jj}\!}{\ZR g{\hh}}\label{eq:igj}\\
\IR g{\jj} & \cong\fact{\IR g{\hh}\!\ZR{\hh}{\jj}\!}{\ZR{\hh}{\jj}}\label{eq:igjzjh}\\
\shortintertext{and}\fact{\IR{\hh}{\jj}\!}{\IR g{\jj}} & \cong\fact{\zent{\hh}\!}{\IR g{\hh}\!\ZR{\hh}{\jj}}\label{eq:zgjzgh}
\end{align}
\end{lem}
\begin{proof}
Consider the restriction of $\rex g{\hh}$ to $\ZR g{\jj}$. By \prettyref{lem:comp},
its image equals $\IR g{\hh}\cap\ZR{\hh}{\jj}$, while its kernel
is $\ZR g{\hh}$, and \prettyref{eq:igj} follows from the homomorphism
theorem. Next, denote by $\phi$ the restriction of $\rex{\hh}{\jj}$
to $\IR g{\hh}$: since, once again by \prettyref{lem:comp}, one
has $\im{\phi}\!=\!\rex{\hh}{\jj}\!\left(\IR g{\hh}\right)\!=\!\IR g{\jj}$
and $\ke{\phi}\!=\!\rex g{\hh}\!\left(\ZR g{\jj}\right)\!=\!\IR g{\hh}\cap\ZR{\hh}{\jj}$,
hence $\IR g{\jj}\!\cong\!\fact{\IR g{\hh}\!}{\IR g{\hh}\cap\ZR{\hh}{\jj}}$,
and \prettyref{eq:igjzjh} follows from this by substituting $A\!=\!\IR g{\hh}$
and $B\!=\!\ZR{\hh}{\jj}$ into \prettyref{eq:1stiso}. Finally, \prettyref{eq:zgjzgh}
follows from \prettyref{eq:igjzjh} by substituting $A\!=\!\ZR{\hh}{\jj}$,
$B\!=\!\IR g{\hh}\!\ZR{\hh}{\jj}$ and $C\!=\!\zent{\hh}$ into \prettyref{eq:2ndiso}
\[
\fact{\IR{\hh}{\jj}\!}{\IR g{\jj}}\!\cong\!\fact{\left(\fact{\zent{\hh}\!}{\ZR{\hh}{\jj}}\right)\!}{\!\left(\fact{\IR g{\hh}\!\ZR{\hh}{\jj}\!}{\ZR{\hh}{\jj}}\right)}\!\cong\!\fact{\zent{\hh}\!}{\IR g{\hh}\!\ZR{\hh}{\jj}}
\]
and taking into account that $\IR{\hh}{\jj}\!\cong\!\fact{\zent{\hh}}{\ZR{\hh}{\jj}}$
by \prettyref{eq:seq1}.
\end{proof}
\global\long\def\jj{\mathfrak{h}_{\mathrm{1}}}%
\global\long\def\hh{\mathfrak{h}_{\mathrm{2}}}%

\begin{lem}
\label{lem:compZR}$\ZR g{\jj}\cap\ZR g{\hh}=\ZR g{\jj\mathrm{\vee\,}\hh}$
and $\ZR g{\jj}\ZR g{\hh}\!<\!\ZR g{\jj\!\cap\!\hh}$ in case $\mathfrak{\jj},\hh\!\subseteq\!\mathfrak{g}$,
and in particular $\ZR g{h\!\vee\!\mzq g}\!=\!\ZR gh$ and $\ZR g{h\!\cap\!\mzq g}\!=\!\zent g$
for any $\mathfrak{h}\!\subseteq\!\mathfrak{g}$.
\end{lem}
\begin{proof}
$\zcl\!\in\!\zent g$ belongs to $\ZR g{\jj}\!\cap\!\ZR g{\hh}$ if
it is contained in both $\du{\jj}$ and $\du{\hh}$, i.e. $\zcl\!\subseteq\!\du{\jj}\!\cap\du{\hh}=\!\du{\left(\jj\mathrm{\vee\,}\hh\right)}$,
proving the first claim. But according to \prettyref{lem:compcor}
the inclusions $\jj\!\cap\!\hh\!\subseteq\!\mathfrak{\jj},\hh\!\subseteq\!\mathfrak{g}$
imply that both $\ZR g{\jj}$ and $\ZR g{\hh}$ are subgroups of $\ZR g{\jj\!\cap\!\hh}$,
hence $\ZR g{\jj}\ZR g{\hh}\!<\!\ZR g{\jj\!\cap\!\hh}$. The final
claim follows from the above by substituting $\jj\!=\!\mathfrak{h}$
and $\hh\!=\!\mzq g$, and taking into account $\ZR g{\mzq g}\!=\!\zent g$.
\end{proof}

\section{The Galois correspondence\label{sec:Galois}}

Let's introduce the notation $\ext hg$ to indicate that $\mathfrak{h\!\in\!\lat}$
is a central quotient of $\mathfrak{g}\!\in\!\lat$ (equivalently,
that $\mathfrak{g}$ is a central extension of $\mathfrak{h}$); in
other words, $\ext hg$ means $\mzq g\!\subseteq\!\mathfrak{h}\!\subseteq\!\mathfrak{g}$
or, what is the same, $\mathfrak{h}\!\subseteq\!\mathfrak{g}\!\subseteq\!\cov{\mathfrak{h}}$
(c.f. \prettyref{sec:Central-quotients-and}). Note that this is not
a transitive relation, i.e. $\ext hj$ and $\ext jg$ does not necessarily
imply $\ext hg$. On the other hand, because every central quotient
is a central extension of the maximal central quotient,  $\ext hg$
implies $\ext{\mzq g}h$, and this gives in turn $\ext{\mzq h}{\mzq g}$.

\begin{thm}
\label{thm:resim} $\usub{\ZR gh}\!=\!\du h$ and $\usub{\IR gh}\!=\!\ZD g$
in case $\ext hg$, hence\textup{ $\mathfrak{h}\!=\!\fact{\mathfrak{g}}{\ZR gh}$
}\textup{\emph{and}}\textup{  $\fact{\mathfrak{h}}{\IR gh\!=\!\mzq g}$;}\textup{\emph{
in particular,}}\textup{ 
\begin{equation}
\IR gh\!=\!\ZR h{\mzq g}\label{eq:fund}
\end{equation}
}
\end{thm}
\begin{proof}
If $\mathfrak{h}$ is a central quotient of $\mathfrak{g}$, then
$\mathfrak{h}\!=\!\fact{\mathfrak{g}}G$ for some $G\!<\!\zent g$
with $\usub G\!=\!\du h$, c.f. \prettyref{sec:Central-quotients-and},
hence all $\mathfrak{g}$-classes contained in $\du h$ are central
and $\usub{\ZR gh}\!=\!\du h$, consequently $G\!=\!\ZR gh$. But
$\ext hg$ implies $\ext{\mzq g}h$, hence $\mzq g\!=\!\fact{\mathfrak{h}}H$
for some subgroup $H\!<\!\zent h$, and because $\usub H\!=\!\du{\left(\mzq g\right)}$,
every $\mathfrak{h}$-class contained in $\ZD g$ is central; since
$\ZD g$ is a union of $\mathfrak{h}$-classes, we get that $\usub{\IR gh}\!=\!\du{\left(\mzq g\right)}\!=\!\usub{\ZR h{\mzq g}}$,
proving that indeed $\IR gh\!=\!\ZR h{\mzq g}$.
\end{proof}
To better grasp the meaning of \prettyref{thm:resim}, let's rewrite
it in a (perhaps) more suggestive form. To this end, consider a subgroup
$H\!<\!\zent g$, and let $\mathfrak{h}\!=\!\zquot gH$, so that $\ext hg$.
Since $\ext hg$ implies $\ext{\mzq h}{\mzq g}$, there exists a subgroup
$\boldsymbol{\partial}H\!<\!\zent{\mzq g}$ such that $\mzq h\!=\!\zquot{\mzq g}{\mathit{\boldsymbol{\partial}H}}$,
and by \prettyref{thm:resim} one has $\ZR gh\!=\!H$ and $\ZR{\mzq g}{\mzq h}\!=\!\boldsymbol{\partial}H$.
It follows from \prettyref{eq:seq1} that there is a short exact sequence
\begin{equation}
\begin{CD}1@>>>\ZR gh@>>>\zent g@>>>\IR gh@>>>1\end{CD}\label{eq:seq1A}
\end{equation}
As $\ext hg$ implies $\ext{\mzq g}h$, we can substitute $\mzq g$
for $\mathfrak{h}$, and simultaneously $\mathfrak{h}$ for $\mathfrak{g}$
in \prettyref{eq:seq1A} to yield the exact sequence
\begin{equation}
\minCDarrowwidth26pt\begin{CD}1@>>>\ZR h{\mzq g}@>>>\zent h@>>>\IR h{\mzq g}@>>>1\end{CD}\label{eq:seq1B}
\end{equation}

But $\IR h{\mzq g}\!=\!\ZR{\mzq g}{\mzq h}\!=\!\boldsymbol{\partial}H$
and $\ZR h{\mzq g}\!=\!\IR gh\!\cong\!\fact{\zent g\!}{\ZR gh}\!=\!\fact{\zent g\!}{\!H}$
by \prettyref{thm:resim} and \prettyref{eq:seq1}, hence we arrive
finally at the exact sequence
\begin{equation}
\minCDarrowwidth29pt\begin{CD}1@>>>\fact{\zent g\!}{\!H}@>>>\zent{\zquot g{\mathit{H}}}@>>>\boldsymbol{\partial}H@>>>1\end{CD}\label{eq:seq1c}
\end{equation}
which shows clearly that $\boldsymbol{\partial}H$ measures the extent
by which $\zent{\zquot g{\mathit{H}}}$ differs from the naive expectation
$\fact{\zent g\!}{\!H}$. In particular, one has $\zent{\zquot g{\mathit{H}}}\!\cong\!\fact{\zent g\!}{\!H}$
only in case $\boldsymbol{\partial}H\!=\!1$, i.e. when $\mzq g\!=\!\mzq h$.

\global\long\def\jj{\mathfrak{h}}%
\global\long\def\hh{\mathfrak{j}}%

\begin{lem}
\label{lem:quotincl}If $\ext hg$ and $\mathfrak{h}\!\subseteq\!\mathfrak{j}\!\subseteq\!\mathfrak{g}$,
then $\ZR jh\!<\!\IR gj$, consequently
\begin{align}
\ZR{\hh}{\jj} & \cong\fact{\ZR g{\jj}\!}{\ZR g{\hh}}\label{eq:igj2}\\
\shortintertext{and}\ZR{\mzq g}{\mzq j} & \cong\fact{\IR{\hh}{\jj}\!}{\IR g{\jj}}\label{eq:ijh}
\end{align}
\end{lem}
\begin{proof}
$\ext hg$ and $\mathfrak{h}\!\subseteq\!\mathfrak{j}\!\subseteq\!\mathfrak{g}$
imply $\ext jg$, hence $\ext{\mzq g}j$, consequently $\IR gj\!=\!\ZR j{\mzq g}$
and $\IR j{\mzq g}\!=\!\ZR{\mzq g}{\mzq j}$ according to \prettyref{thm:resim}.
On the other hand, $\mzq g\!\subseteq\!\mathfrak{h}\!\subseteq\!\mathfrak{j}$
gives $\ZR jh\!<\!\ZR j{\mzq g}$ according to \prettyref{lem:comp},
hence $\ZR jh\!<\!\ZR j{\mzq g}\!=\!\IR gj$ by the above. But $\ZR jh\!<\!\IR gj$
implies  $\IR g{\hh}\!\cap\!\ZR{\hh}{\jj}\!=\!\ZR{\hh}{\jj}$ and
$\IR g{\hh}\!\ZR{\hh}{\jj}\!=\!\IR g{\hh}$, and taking this into
account, \prettyref{eq:igj} reduces to \prettyref{eq:igj2}, while
\prettyref{eq:zgjzgh} leads to 
\[
\fact{\IR{\hh}{\jj}\!}{\IR g{\jj}}\!\cong\!\fact{\zent j\!}{\IR g{\hh}}\!=\!\fact{\zent j\!}{\ZR j{\mzq g}}\!\cong\!\IR j{\mzq g}\!=\!\ZR{\mzq g}{\mzq j}
\]
in view of \prettyref{thm:resim} and \prettyref{eq:seq1}, proving
\prettyref{eq:ijh}.
\end{proof}
For subgroups $H_{1}\!<\!H_{2}\!<\!\zent g$, \prettyref{eq:igj2}
with the choice $\mathfrak{j}\!=\!\fact{\mathfrak{g}}{\mathit{\!H_{\mathrm{1}}}}$
and $\mathfrak{h}\!=\!\fact{\mathfrak{g}}{\!\mathit{H_{\mathrm{2}}}}$
gives us an isomorphism
\begin{equation}
\ZR{\fact g{\mathit{\!H_{\mathrm{1}}}}}{\fact g{\!\mathit{H_{\mathrm{2}}}}}\!\cong\!\fact{H_{2}}{\!H_{1}}\label{eq:quotZ}
\end{equation}
while combining \prettyref{eq:ijh} with \prettyref{eq:seq1A} leads
to the exact sequence
\begin{equation}
\minCDarrowwidth23pt\begin{CD}1@>>>\fact{\zent g\!}{\!H_{2}}@>>>\IR{\fact g{\mathit{\!H_{\mathrm{1}}}}}{\fact g{\!\mathit{H_{\mathrm{2}}}}}@>>>\boldsymbol{\partial}H_{1}@>>>1\end{CD}\label{eq:4seqjh}
\end{equation}

To formulate our next result, let's introduce the notation $\inn gh\!=\!\set{\mathfrak{j}\!\in\!\lat}{\mathfrak{h}\!\subseteq\!\mathfrak{j}\!\subseteq\!\mathfrak{g}}$
for $\mathfrak{g},\mathfrak{h}\in\lat$; notice that $\inn gh$ is
empty unless $\mathfrak{h}\!\subseteq\!\mathfrak{g}$. In particular,
$\inn g{\mzq g}$ equals the collection of central quotients, and
$\inn{\cov g}g$ that of central extensions of $\mathfrak{g}$ according
to the comments following \prettyref{eq:usubzent}. As a byproduct
of \prettyref{thm:resim}, one gets an interesting extension of the
Galois correspondence that has been described in \prettyref{sec:Central-quotients-and},
according to which there is, in case $\mathfrak{h}$ is a central
quotient of $\mathfrak{g}$, a one-to-one correspondence between the
subgroups of $\ZR gh$ and those central quotients of $\mathfrak{g}$
that contain $\mathfrak{h}$. But $\ext hg$ implies $\ext{\mzq g}h$,
i.e. $\mzq g$ is a central quotient of $\mathfrak{h}$, hence there
is a one-to-one correspondence between the subgroups of $\ZR h{\mzq g}\!=\!\IR gh$
and those central quotients of $\mathfrak{h}$ that contain $\mathfrak{\mzq g}$.
Since these are nothing but the central quotients of $\mathfrak{g}$
contained in $\mathfrak{h}$, in the end we get a one-to-one correspondence
between the latter and the subgroups of $\IR gh$. For later reference,
let's formulate this result as a Lemma.
\begin{lem}
\label{lem:galcor}In case $\ext hg$, there are one-to-one correspondences
between $\inn gh$ and subgroups of $\ZR gh$ on one hand, and between
$\inn h{\mzq g}$ and subgroups of $\IR gh$ on the other.
\end{lem}
Combined with the modularity of the lattice $\lat$, \prettyref{lem:galcor}
leads to important results, notably the following strengthening of
\prettyref{lem:compZR}.

\global\long\def\jj{\mathfrak{h}_{\mathrm{1}}}%
\global\long\def\hh{\mathfrak{h}_{\mathrm{2}}}%

\begin{lem}
\label{lem:compZR*} $\ZR g{\jj}\ZR g{\hh}\!=\!\ZR g{\jj\!\cap\!\hh}$
if both $\jj$ and $\hh$ are central quotients of $\mathfrak{g}$.
\end{lem}
\begin{proof}
Indeed, the inclusions $\mzq g\!\subseteq\!\mathfrak{\jj},\hh\!\subseteq\!\mathfrak{g}$
imply $\mzq{\jj}\!\subseteq\!\mzq g\!\subseteq\!\jj\!\cap\!\hh\!\subseteq\!\jj$
and $\mzq{\left(\jj\mathrm{\vee\,}\hh\right)}\!\subseteq\!\mzq g\!\subseteq\!\hh\!\subseteq\!\jj\mathrm{\vee\,}\hh$,
hence $\ext{\jj\!\cap\!\hh}{\jj}$ and $\ext{\hh}{\jj\mathrm{\vee\,}\hh}$,
so there is, by \prettyref{lem:galcor}, a one-to-one correspondence
between $\inn{\jj\mathrm{\vee\,}\hh}{\hh}$ and subgroups of $\ZR{\jj\mathrm{\vee\,}\hh}{\hh}$
on one hand, and between $\inn{\jj}{\jj\!\cap\!\hh}$ and subgroups
of $\ZR{\jj}{\jj\!\cap\!\hh}$ on the other. $\ZR g{\jj}\ZR g{\hh}$
is a subgroup of $\ZR g{\jj\!\cap\!\hh}$ by \prettyref{lem:compZR}:
should it be a proper subgroup, the factor group $\fact{\ZR g{\jj}\ZR g{\hh}\!}{\ZR g{\jj}}$,
which is isomorphic to $\fact{\ZR g{\hh}\!}{\ZR g{\jj\mathrm{\vee\,}\hh}}\!\cong\!\ZR{\jj\mathrm{\vee\,}\hh}{\hh}$
according to \prettyref{eq:1stiso} and \prettyref{lem:quotincl},
would have less subgroups than the factor group $\fact{\ZR g{\jj\!\cap\!\hh}\!}{\ZR g{\jj}\!\cong\!\ZR{\jj}{\jj\!\cap\!\hh}}$,
contradicting the one-to-one correspondence between $\inn{\jj\mathrm{\vee\,}\hh}{\hh}$
and $\inn{\jj}{\jj\!\cap\!\hh}$ that follows from the modularity
of $\lat$.
\end{proof}
The above result has the following important consequence: consider
subgroups $H_{1},H_{2}\!<\!\zent g$, and let $\jj\!=\!\fact{\mathfrak{g}}{\!H_{1}}$
and $\hh\!=\!\fact{\mathfrak{g}}{\!H_{2}}$, so that $H_{1}\!=\!\ZR g{\jj}$
and $H_{2}\!=\!\ZR g{\hh}$. Since $H_{1}\!\cap\!H_{2}\!=\!\ZR g{\jj\!\vee\!\hh}$
by \prettyref{lem:compZR}, while $H_{1}H_{2}\!=\!\ZR g{\jj\!\cap\!\hh}$
according to \prettyref{lem:compZR*}, one has by \prettyref{thm:resim}
\begin{align}
\left(\zquot g{H_{1}}\right)\!\vee\!\left(\zquot g{H_{2}}\right)\! & =\!\zquot g{\left(H_{1}\!\cap\!H_{2}\right)}\label{eq:latisoma}\\
\shortintertext{and}\left(\zquot g{H_{1}}\right)\!\cap\!\left(\zquot g{H_{2}}\right) & \!=\!\zquot g{\left(H_{1}H_{2}\right)}\label{eq:latisomb}
\end{align}
which means that the correspondences of \prettyref{lem:galcor} are
actually lattice isomorphisms.

\global\long\def\jj{\mathfrak{h}}%
\global\long\def\hh{\mathfrak{j}}%

In the same vein, one gets the following far-reaching generalization
of \prettyref{thm:resim}.
\begin{thm}
\label{thm:igh}$\fact{\mathfrak{g}}{\ZR gh}\!=\!\mathfrak{h}\!\vee\!\mzq g$
and $\fact{\mathfrak{h}}{\IR gh}\!=\!\mathfrak{h}\!\cap\!\mzq g$
for $\mathfrak{h}\!\subseteq\!\mathfrak{g}$, and in particular
\begin{equation}
\IR gh\!=\!\ZR h{h\!\cap\!\mzq g}\label{eq:fund*}
\end{equation}
\end{thm}
\begin{proof}
As $\mathfrak{h}\!\subseteq\!\mathfrak{g}$ implies both $\mzq g\!\subseteq\!\mathfrak{h}\!\vee\!\mzq g\!\subseteq\!\mathfrak{g}$
and $\mzq h\!\subseteq\!\mathfrak{h}\!\cap\!\mzq g\!\subseteq\!\mathfrak{h}$,
we have $\fact{\mathfrak{g}}{\ZR g{h\!\vee\!\mzq g}}\!=\!\mathfrak{h}\!\vee\!\mzq g$
and $\fact{\mathfrak{h}}{\ZR h{h\!\cap\!\mzq g}}\!=\!\mathfrak{h}\!\cap\!\mzq g$
according to \prettyref{thm:resim}. But $\ZR g{h\!\vee\!\mzq g}\!=\!\ZR gh$
by \prettyref{lem:compZR}, proving the first assertion, while the
second one would follow from the equality $\IR gh\!=\!\ZR h{h\!\cap\!\mzq g}$.
To prove the latter, consider $\mathfrak{j}\!=\!\fact{\mathfrak{h}}{\IR gh}$;
we claim that $\mathfrak{j}\!\subseteq\!\mathfrak{h}\!\cap\!\mzq g$.
The inclusion $\mathfrak{j}\!\subseteq\!\mathfrak{h}$ is obvious,
while $\mathfrak{j}\!\subseteq\!\mzq g$ follows from the observation
that, because every central $\mathfrak{g}$-class is contained in
a central $\mathfrak{j}$-class, the union $\usub{\zent g}\!=\!\du{\left(\mzq g\right)}$
of all central $\mathfrak{g}$-classes should be contained in the
union $\usub{\IR gj}\!=\!\du j$ of all those central $\mathfrak{j}$-classes
that contain at least one central $\mathfrak{g}$-class, hence $\du{\left(\mzq g\right)}\!\subseteq\!\du j$
or, what is the same, $\mathfrak{j}\!\subseteq\!\mzq g$. But $\mathfrak{j}\!\subseteq\!\mathfrak{h}\!\cap\!\mzq g$
means that $\fact{\mathfrak{h}}{\IR gh}\!\subseteq\!\fact{\mathfrak{h}}{\ZR h{h\!\cap\!\mzq g}}$,
hence $\ZR h{h\!\cap\!\mzq g}$ is a subgroup of $\IR gh$ by \prettyref{lem:galcor}.
We claim that they are actually equal.

According to \prettyref{lem:galcor}, there is a one-to-one correspondence
between subgroups of $\IR g{\mathfrak{h}\!\vee\!\mzq g}$ and $\inn{\mathfrak{h}\!\vee\!\mzq g}{\mzq g}$
on one hand, and between subgroups of $\ZR h{h\!\cap\!\mzq g}$ and
$\inn{\mathfrak{h}}{\mathfrak{h}\!\cap\!\mzq g}$ on the other. But
$\IR g{\mathfrak{h}\!\vee\!\mzq g}\!\cong\!\fact{\zent g\!}{\ZR g{h\!\vee\!\mzq g}}\!=\!\fact{\zent g\!}{\ZR gh}\!\cong\!\IR gh$
by \prettyref{eq:seq1}: should $\ZR h{h\!\cap\!\mzq g}$ be a proper
subgroup of $\IR gh$, it would have less subgroups, contradicting
the one-to-one correspondence between $\inn{\mathfrak{h}}{\mathfrak{h}\!\cap\!\mzq g}$
and $\inn{\mathfrak{h}\!\vee\!\mzq g}{\mzq g}$ that follows from
the modularity of the deconstruction lattice.
\end{proof}
\begin{singlespace}
\begin{center}
\usetikzlibrary {shapes.geometric,shapes.symbols}
\begin{tikzpicture}[auto,node distance=1.4cm,
every loop/.style={}, 
main node/.style={},
null/.style={coordinate}]
\node[main node] (1) {$\mathfrak{g}$}; 
\node[main node] (2)[below of=1]{$\quad\thinspace\thinspace\mathfrak{h}\negthinspace\vee\negthinspace\mzq{g}\negthinspace= \fact{\mathfrak{g}}{\ZR gh}$}; 
\node[null] (3)  [below  of=2] {};
\node[main node] (4) [below right of=2] {$\thinspace\mzq g$};  
\node[null] (5)  [below left of=4] {};
\node[main node] (6) [left of=5] {$\mathfrak{h}$}; 
\node[null] (7)  [below right of=6] {}; 
\node[main node] (8) [below right  of=6] {$\quad\mathfrak{h}\negthinspace\cap\negthinspace\mzq{g}\negthinspace= \fact{\mathfrak{h}}{\IR gh}$};
\node[main node] (9) [below  of=8] {$\mzq h$\quad};
\path[every node/.style={font=\sffamily\small}]        
(1) edge node {} (2) 
(2) edge node {} (6)  
edge node {} (4)
(6) edge node {} (8)
(4) edge node {} (8)
(8) edge node {} (9)
; 
\end{tikzpicture}
\par\end{center}

\begin{center}
\medskip{}
Fig.3: Inclusion relations relevant to \prettyref{thm:igh}.
\par\end{center}
\end{singlespace}

\medskip{}

\section{A long exact sequence\label{sec:Long-exact-sequences}}

An interesting addendum to the previous results follows from the observation
that, thanks to \prettyref{eq:fund}, the short exact sequences \prettyref{eq:seq1A}
and \prettyref{eq:seq1B} can be combined for $\ext hg$ into a four-term
exact sequence, and the latter leads, by means of a recursive process,
to a long exact sequence connecting the centers of higher central
quotients. 

Too see how this come about, remember that $\IR gh\!=\!\ZR h{\mzq g}$
for $\ext hg$ by \prettyref{thm:resim}, hence the short exact sequence
of \prettyref{eq:seq1A} reads in this case 
\begin{equation}
\begin{CD}1@>>>\ZR gh@>>>\zent g@>>>\ZR h{\mzq g}@>>>1\end{CD}\label{eq:seq2}
\end{equation}
Next, since $\ext hg$ implies $\ext{\mzq g}h$, we can substitute
$\mzq g$ for $\mathfrak{h}$, and simultaneously $\mathfrak{h}$
for $\mathfrak{g}$ in \prettyref{eq:seq2} to yield
\begin{equation}
\begin{CD}1@>>>\ZR h{\mzq g}@>>>\zent h@>>>\ZR{\mzq g}{\mzq h}@>>>1\end{CD}\label{eq:seq3}
\end{equation}
and combining this last sequence with \prettyref{eq:seq2} gives the
four-term exact sequence
\begin{equation}
\minCDarrowwidth29pt\begin{CD}~~~~~~1@>>>\ZR gh@>>>\zent g@>>>\zent h@>>>\ZR{\mzq g}{\mzq h}@>>>1\end{CD}\label{eq:seq4}
\end{equation}

But there is no reason to stop here, as $\ext hg$ implies $\ext{\mzn{\mathit{k}}h}{\mzn{\mathit{k}}g}$
for any $k\!\geq\!1$, so we can substitute $\mzn{\mathit{k}}g$ for
$\mathfrak{g}$ and $\mathfrak{\mzn{\mathit{k}}h}$ for $\mathfrak{h}$
in \prettyref{eq:seq4} to obtain an exact sequence
\[
\minCDarrowwidth14pt\begin{CD}1@>>>\ZR{\mzn{\mathit{k}}g}{\mzn{\mathit{k}}h}@>>>\zent{\mzn{\mathit{k}}g}@>>>\zent{\mzn{\mathit{k}}h}@>>>\ZR{\mzn{\mathit{k}+1}g}{\mzn{\mathit{k}+1}h}@>>>1\end{CD}
\]
Combining \prettyref{eq:seq4} with the above sequences for $k\!=\!1,\ldots,n$
leads finally to the long exact sequence
\begin{flushleft}
~~~~~~~\begin{tikzcd}[sep=small]1\arrow[r]&\ZR gh\arrow{r}&\zent{g}\arrow{r}\arrow[d,phantom,""{coordinate,name=Z}]&\zent{h}\arrow{r}&\zent{\mzq g}\arrow{r}&\zent{\mzq h}
\arrow[dllll,rounded corners,to path={--([xshift=1ex]\tikztostart.east)|-(Z)[near end]\tikztonodes-|([xshift=-4ex]\tikztotarget.west)--(\tikztotarget)}]\\ &\cdots\arrow[r]&\zent{\mzn {\mathit{n}}g}\arrow{r}&\zent{\mzn {\mathit{n}}h}\arrow[r]&\ZR{\mzn {\mathit{n}+1}g}{\mzn {\mathit{n}+1}h}\arrow[r]&1
\end{tikzcd}
\par\end{flushleft}

Since $\mzn{\mathit{n}+1}g=\mzn{\mathit{n}}g$ for large enough $n$
(c.f. \prettyref{sec:Central-quotients-and}) and $\mzn{\mathit{k}+1}g\subseteq\mzn{\mathit{k}}h\subseteq\mzn{\mathit{k}}g$
for $\ext hg$ and $k\!>\!1$, there exists a largest integer $N$
such that $\mzn{\mathit{N}}h\!\neq\!\mzn{\mathit{N}}g$ (clearly,
this integer depends on both $\mathfrak{g}$ and $\mathfrak{h}$).
This leads finally to the following result.
\begin{thm}
\label{thm:lexseq}If $\ext hg$ and $N$ is the largest integer such
that $\mzn{\mathit{N}}h\!\neq\!\mzn{\mathit{N}}g$, then there is
a long exact sequence
\begin{flushleft}
\smallskip{}
~~~~~~~~\begin{tikzcd}[]\; 1\arrow[r]&\ZR{\mathfrak{g}}{\mathfrak{h}}\arrow{r}\arrow[d,phantom,""{coordinate,name=Z}]&\zent{\mathfrak{g}}\arrow{r}&\zent{\mathfrak{h}}
\arrow{r}&\zent{\mzq g}
\arrow[dlll,rounded corners,to path={--([xshift=3ex]\tikztostart.east)|-(Z)[near end]\tikztonodes-|([xshift=-3ex]\tikztotarget.west)--(\tikztotarget)}]\\ &\zent{\mzq h}\arrow[r]&\cdots\arrow[r]&\zent{\mzn {\mathit{N}}g}\arrow{r}&\zent{\mzn {\mathit{N}}h}\arrow[r]&1
\end{tikzcd}\medskip{}
\par\end{flushleft}
\end{thm}
\begin{proof}
This is a direct consequence of the above considerations, taking into
account that $\ZR{\mzn{\mathit{N}+1}g}{\mzn{\mathit{N}+1}h}\!=\!1$
since $\mzn{\mathit{N}+1}h\!=\!\mzn{\mathit{N}+1}g$.
\end{proof}

\section{Extensions vs quotients\label{sec:Extensions-vs-quotients}}

We have gone a long way in achieving our original goal of understanding
the relations between the centers of different elements of $\lat$,
c.f. the exact sequences \prettyref{eq:seq1c} and \prettyref{eq:4seqjh},
or the isomorphism \prettyref{eq:quotZ}. But this is by no means
the end of the story, as practical considerations suggest that one
should also consider the case of central extensions besides that of
quotients. Of course, since $\ext hg$ not only means that $\mathfrak{h}$
is a central quotient of $\mathfrak{g}$, but also that $\mathfrak{g}$
is a central extension of $\mathfrak{h}$, all the results of \prettyref{sec:Galois}
apply, but they have the drawback of referring to central quotients
of $\mathfrak{g}$ and $\mathfrak{h}$ instead of their central extensions.
Fortunately, this can be remedied by using the following result. 
\begin{lem}
\label{lem:dualZR}$\ZR{\du h}{\du g}\!\cong\!\ZR gh$ and $\IR{\du h}{\du g}\!\cong\!\ZR{\cov h}g$
provided $\ext hg$.
\end{lem}
\begin{proof}
First, note that $\ext hg$ iff $\ext{\du g}{\du h}$. Since by assumption,
$\mathfrak{h}\!=\!\fact{\mathfrak{g}}H$ for some subgroup $H\!<\!\zent g$,
hence \cite{Bantay2020a} there is a subgroup $\text{\ensuremath{\du{\mathit{H}}}}$
of $\zent{\du h}$ isomorphic with $H$ such that $\du g\!=\!\fact{\du h}{\text{\ensuremath{\du{\mathit{H}}}}}$,
and this implies the isomorphism $\ZR{\du h}{\du g}\!=\!\du{\mathit{H}}\!\cong\!H\!=\!\ZR gh$
by \prettyref{thm:resim}. As to the second assertion, it follows\prettyref{eq:fund}
and the above since $\IR{\du h}{\du g}\!=\!\ZR{\du g}{\mzq{\left(\du h\right)}}\!=\!\ZR{\du g}{\du{\left(\cov h\right)}}\!\cong\!\ZR{\cov h}g$,
proving the claim.
\end{proof}
In particular, substituting $\du h$ for $\mathfrak{g}$ and $\du g$
for $\mathfrak{h}$ in \prettyref{eq:seq1A}, and making use of \prettyref{lem:dualZR}
leads to the short exact sequence
\begin{equation}
\begin{CD}1@>>>\ZR gh@>>>\zent{\du h}@>>>\ZR{\cov h}g@>>>1\end{CD}\label{eq:dseq2}
\end{equation}
to be contrasted with \prettyref{eq:seq2}, while a similar argument
applied to \prettyref{eq:seq4} leads to the exact sequence
\begin{equation}
\minCDarrowwidth18pt\begin{CD}1@>>>\ZR gh@>>>\zent{\du h}@>>>\zent{\du g}@>>>\ZR{\cov g}{\cov h}@>>>1\end{CD}\label{eq:dualseq}
\end{equation}

Let's note that, by making use of \prettyref{thm:igh}, one can drop
the requirement $\ext hg$, resulting in the following generalization
of \prettyref{lem:dualZR}.
\begin{lem}
\label{lem:dualZR*}$\ZR{\du h}{\du g}\!\cong\!\ZR{g\!\cap\!\cov h}h$
and $\IR{\du h}{\du g}\!\cong\!\ZR{\cov h\!\vee\!g}g$ for $\mathfrak{h}\!\subseteq\negmedspace\mathfrak{g}$.
\end{lem}
\begin{proof}
First of all, let's note that $\mzq{\left(\du h\right)}\!\subseteq\!\du g\!\cap\!\mzq{\left(\du h\right)}\!\subseteq\!\du g\!\subseteq\!\du g\!\vee\!\mzq{\left(\du h\right)}\!\subseteq\!\du h$
whenever $\mathfrak{h}\!\subseteq\negmedspace\mathfrak{g}$, and this
implies $\ZR{\du h}{\du g}\!=\!\ZR{\du h}{\du g\!\vee\!\mzq{\left(\du h\right)}}\!\cong\!\ZR{g\!\cap\!\cov h}h$
and $\!\IR{\du h\!}{\du g}\!\!=\!\!\ZR{\du g}{\du g\!\cap\!\mzq{\left(\du h\right)}}\!\cong\!\ZR{g\!\vee\!\cov h}g$
according to \prettyref{thm:igh} and \prettyref{lem:dualZR}.
\end{proof}
Substituting $\du h$ for $\mathfrak{g}$ and $\du g$ for $\mathfrak{h}$
in \prettyref{eq:seq1A} gives (for $\mathfrak{h}\!\subseteq\negmedspace\mathfrak{g}$)
\begin{equation}
\begin{CD}1@>>>\ZR{g\!\cap\!\cov h}h@>>>\zent{\du h}@>>>\ZR{\cov h\!\vee\!g}g@>>>1\end{CD}\label{eq:dseq1}
\end{equation}
when taking into account \prettyref{lem:dualZR*}, to be compared
to the sequence
\begin{equation}
\begin{CD}1@>>>\ZR g{h\vee\mzq g}@>>>\zent g@>>>\ZR h{\mathfrak{h}\!\cap\!\mzq g}@>>>1\end{CD}\label{eq:dseq3}
\end{equation}
that follows from \prettyref{eq:fund*}.

\section{Summary and outlook\label{sec:Summary}}

The main theme of our investigations was to characterize the mutual
relations between the centers of different elements of the deconstruction
lattice, with particular emphasis on the central quotients and extensions
of a given $\mathfrak{g}\!\in\!\lat$. We have seen that one can go
a long way in this direction using exact sequences, as exhibited by
results like \prettyref{eq:seq1c}, \prettyref{thm:lexseq} and \prettyref{eq:dualseq}.
While not providing a direct and constructive description of the relevant
groups, this information is usually enough to pin down (more or less
uniquely) their structure, and this is what matters for most applications.
Once the relevant centers are known, one can use this knowledge to
simplify greatly otherwise cumbersome computations that could prove
difficult in case of large examples. In particular, the explicit Galois
correspondence described by \prettyref{lem:galcor} is a very useful
tool in actual computations, while \prettyref{lem:dualZR*} settles
the dual relation of central quotients and extensions.

Apart from the practical considerations set forth above, there are
more conceptual issues underlying the interest in the previous results.
One such is the search for an analogue of the famous lemma of Gr�n
\cite{Robinson}, according to which the upper central series of a
perfect group has length at most two. In our context this would mean
that, provided $\mathfrak{g}\!\in\!\lat$ does not contain any simple
current, the maximal central quotient $\mzq g$ should be centerless,
i.e. $\zent{\mzq g}\!=\!1$. While this holds automatically for local
$\mathfrak{g}\!\in\!\lat$ corresponding to some finite group, the
generic case seems much more difficult to prove, as there is no clear
adaptation of the group theoretic techniques used in the proof of
Gr�n's lemma. 

Another interesting question concerns the analogue of Ito's theorem
\cite{Isaacs,Serre}, which in our case is tantamount to the claim
that the ratio
\[
\frac{1}{\qd{\alpha}}\sum_{p\in\mzq g}\qd p^{2}
\]
is an algebraic integer for every $\alpha\!\in\!\mathfrak{g}$. Should
this claim hold, it would restrict severely the arithmetic properties
of the quantum dimensions and, more generally, the Galois action on
the primaries \cite{Bantay2020a}. As before, the claim follows from
known group theoretic results for local $\mathfrak{g}\!\in\!\lat$
corresponding to some finite group, but its generalization to arbitrary
$\mathfrak{g}\!\in\!\lat$ is far from being obvious.

Finally, we should note that in the preceding discussions we have
neglected one important aspect, namely that the center $\zent g$
of any given $\mathfrak{g}\!\in\!\lat$ is not simply an abelian group,
but has a natural permutation action on the set of $\mathfrak{g}$-classes
\cite{Bantay2020a}, which is compatible with the restriction map.
This permutation action is far from being arbitrary, as can be seen
most directly on the example of abelian $\mathfrak{g}$, i.e. when
all primaries in $\mathfrak{g}$ are simple currents, since in this
case the $\mathfrak{g}$-classes are in one-to-one correspondence
with the characters of $\zent g$, and the permutation action is regular,
i.e. the action of $\zent g$ on itself by translations. On the other
extreme, if $\mathfrak{g}$ equals the maximal element of the deconstruction
lattice $\lat$, then each $\mathfrak{g}$-class contains precisely
one primary, with central classes containing the simple currents,
and the corresponding permutation action is nothing but the action
of the group of simple currents on the set of primaries that is induced
by the fusion product. This latter action is known to have non-trivial
properties \cite{Bantay2005}, and it is natural to expect that, for
$\mathfrak{g}\!\in\!\lat$ intermediate between these two extremes,
the corresponding permutation action still enjoys some interesting
properties. But this circle of questions lies clearly outside the
scope of the present note.

\appendix
\global\long\def\gvh{\mathfrak{g}\sqcup\mathfrak{h}}%
\global\long\def\decl{\textrm{deconstruction lattice}}%

\section[Appendix]{\label{sec:Appendix}Modularity of the deconstruction lattice}

The following characterization of the join operation in $\!\lat$
is interesting in itself.
\begin{lem}
\label{lem:gvh}For $\mathfrak{g},\mathfrak{h}\!\in\!\lat$, a primary
$p$ belongs to the join $\mathfrak{g}\!\vee\!\mathfrak{h}$ iff there
exists $\alpha\!\in\!\mathfrak{g}$ and $\beta\!\in\!\mathfrak{h}$
such that $N_{\alpha\beta}^{p}>0$.
\end{lem}
\begin{proof}
Let $\gvh\!=\!\set p{N_{\alpha\beta}^{p}\!>\!0\textrm{ for some }\alpha\!\in\!\mathfrak{g}\textrm{ and }\beta\!\in\!\mathfrak{h}}$.
Clearly, both $\mathfrak{g}$ and $\mathfrak{h}$ are contained in
$\gvh$. On the other hand, $p\!\in\!\gvh$ means that there exists
$\alpha\!\in\!\mathfrak{g}$ and $\beta\!\in\!\mathfrak{h}$ such
that $N_{\alpha\beta}^{p}\!>\!0$: since $\mathfrak{g}\!\vee\!\mathfrak{h}$
is the least element of $\lat$ containing both $\mathfrak{g}$ and
$\mathfrak{h}$, we have necessarily $p\!\in\!\mathfrak{g}\!\vee\!\mathfrak{h}$,
hence $\gvh\!\subseteq\!\mathfrak{g}\!\vee\!\mathfrak{h}$. To prove
that they are actually equal, it suffices to show that $\gvh\!\in\!\lat$,
i.e. that $p,q\!\in\!\gvh$ and $N_{pq}^{r}\!>\!0$ implies $r\!\in\!\gvh$.
But this is tantamount to showing that, if there exists $\alpha_{1},\alpha_{2}\!\in\!\mathfrak{g}$
and $\beta_{1},\beta_{2}\!\in\!\mathfrak{h}$ such that $N_{\alpha_{1}\beta_{1}}^{p}\!>\!0$
and $N_{\alpha_{2}\beta_{2}}^{q}\!>\!0$, then for all $r$ such that
$N_{pq}^{r}\!>\!0$ there exists $\alpha_{3}\!\in\!\mathfrak{g}$
and $\beta_{3}\!\in\!\mathfrak{h}$ satisfying $N_{\alpha_{3}\beta_{3}}^{r}\!>\!0$.
By the associativity of the fusion algebra, it follows form our assumptions
that
\begin{gather*}
0<N_{\alpha_{1}\beta_{1}}^{p}N_{\alpha_{2}\beta_{2}}^{q}N_{pq}^{r}\leq\sum_{s,t}N_{\alpha_{1}\beta_{1}}^{s}N_{\alpha_{2}\beta_{2}}^{t}N_{st}^{r}=~~~~\\
\sum_{w,t}N_{\alpha_{1}t}^{w}N_{\beta_{1}w}^{r}N_{\alpha_{2}\beta_{2}}^{t}=\sum_{w,u}N_{\alpha_{1}\alpha_{2}}^{u}N_{\beta_{2}u}^{w}N_{\beta_{1}w}^{r}=\sum_{u,v}N_{\alpha_{1}\alpha_{2}}^{u}N_{\beta_{1}\beta_{2}}^{v}N_{uv}^{r}
\end{gather*}
i.e. there should exist primaries $\alpha_{3}$ and $\beta_{3}$ such
that $N_{\alpha_{1}\alpha_{2}}^{\alpha_{3}}$, $N_{\beta_{1}\beta_{2}}^{\beta_{3}}$
and $N_{\alpha_{3}\beta_{3}}^{r}$ are all positive. But $\mathfrak{g},\mathfrak{h}\!\in\!\lat$
and $N_{\alpha_{1}\alpha_{2}}^{\alpha_{3}}N_{\beta_{1}\beta_{2}}^{\beta_{3}}>0$
implies that $\alpha_{3}\!\in\!\mathfrak{g}$ and $\beta_{3}\!\in\!\mathfrak{h}$,
consequently $r\!\in\!\gvh$.
\end{proof}
The previous result allows for the following streamlined proof of
the modularity of $\lat$ (recall that $N_{pq}^{\overline{r}}\!=\!N_{pr}^{\overline{q}}$,
with $\overline{p}$ denoting the charge conjugate of $p$).
\begin{thm}
\label{thm:modularity}The lattice $\lat$ is modular, that is
\[
\mathfrak{h}_{2}\cap\left(\mathfrak{h}_{1}\vee\mathfrak{g}\right)\subseteq\mathfrak{h}_{1}\vee\left(\mathfrak{h}_{2}\cap\mathfrak{g}\right)
\]
 for $\mathfrak{g},\mathfrak{h}_{1},\mathfrak{h}_{2}\!\in\!\lat$
such that $\mathfrak{h}_{1}\!\subseteq\!\mathfrak{h}_{2}$.
\end{thm}
\begin{proof}
Suppose that $\alpha\!\in\!\mathfrak{h}_{2}\cap\left(\mathfrak{h}_{1}\vee\mathfrak{g}\right)$.
Then $\alpha\!\in\!\mathfrak{h}_{2}$, hence $\overline{\alpha}\!\in\!\mathfrak{h}_{2}$,
and by \prettyref{lem:gvh} there exists primaries $\beta\!\in\!\mathfrak{h}_{1}$
and $\gamma\!\in\!\mathfrak{g}$ such that $N_{\beta\gamma}^{\alpha}\!>\!0$.
But $\beta\!\in\!\mathfrak{h}_{2}$ since $\mathfrak{h}_{1}\!\subseteq\!\mathfrak{h}_{2}$,
consequently $N_{\beta\overline{\alpha}}^{\overline{\gamma}}\!=\!N_{\beta\gamma}^{\alpha}$
implies $\overline{\gamma}\!\in\!\mathfrak{h}_{2}$, i.e. $\gamma\!\in\!\mathfrak{h}_{2}$.
All in all, we get that $N_{\beta\gamma}^{\alpha}\!>\!0$ with $\beta\!\in\!\mathfrak{h}_{1}$
and $\gamma\!\in\!\mathfrak{g\cap\mathfrak{h}}_{2}$, consequently
$\alpha\!\in\!\mathfrak{h}_{1}\!\vee\!\left(\mathfrak{h}_{2}\cap\mathfrak{g}\right)$
by \prettyref{lem:gvh}. 
\end{proof}
Modularity has many important consequences, e.g. the Kurosh-Ore theorem
\cite{Gratzer2011}, but we make mainly use of the so-called diamond
isomorphism theorem: for any two elements $a,b$ of a modular lattice
$\left(L,\wedge,\vee\right)$, there is an order-preserving one-to-one
correspondence between the sets $\set{x\!\in\!L}{a\wedge b\leq x\leq a}$
and $\set{x\!\in\!L}{b\leq x\leq a\vee b}$.

\end{document}